\documentclass[12pt]{iopart}

\usepackage{iopams} 
\usepackage{amsthm}
\newtheorem{theorem}{Theorem}
\newtheorem{dfn}{Definition}

\newtheorem{prop}{Proposition}

\newtheorem{remark}{Remark}
\begin{document}

\title[Non-hermitian Hamiltonian from Higgs algebra]{Construction of a new three boson non-hermitian Hamiltonian associated to deformed Higgs algebra : real eigenvalues and Partial ${\mathcal {PT}}$-symmetry }

\author{Arindam Chakraborty}

\address{Department of Physics, Heritage Institute of Technology, Kolkata-700107, India}
\ead{arindam.chakraborty@heritageit.edu}
\vspace{10pt}
\begin{indented}
\item[]November 2021
\end{indented}

\begin{abstract}
A $\gamma$-deformed version of $\mathfrak{su}(2)$ algebra 
has been obtained from a bi-orthogonal system of vectors in $\bf{C^2}$.  Fusion of Jordan-Schwinger realization of complexified $\mathfrak{su}(2)$ with Dyson-Maleev representation gives a 3-boson realization of Higgs algebra of cubic polynomial type. The non-hermitian Hamiltonian thus obtained is found to have real eigenvalues and  eigen states with symmetry induced orthogonality. The notion of partial ${\mathcal {PT}}$-symmetry (henceforth $\partial_{\mathcal { PT}}$) has been introduced as a characteristic feature of these multi-boson realizations.  The Hamiltonian along with its eigenstates have been studied in the light of $\partial_{\mathcal { PT}}$-symmetry. The possibility
of $\partial_{\mathcal { PT}}$-symmetry breaking is also discussed. The deformation parameter $\gamma$ plays a crucial role in the entire formulation and non-trivially modifies the eigenfunctions under consideration. 
\end{abstract}

\noindent{\it Keywords}: Non-hermitian operator, bi-orthogonal vectors, Jordan-Schwinger map, Dyson-Maleev realization, Higgs algebra, Partial ${\mathcal {PT}}$-symmetry.

%\maketitle

\section{Introduction}
The present communication attempts to discuss three interrelated issues regarding multi-boson realization of deformed linear and non-linear (polynomial) algebras.

I. Implementing the idea of boson realization to a deformed version of $\mathfrak{su}(2)$ algebra involving non-hermitian generators obtained from a bi-orthogonal system of vectors. Subsequent complexification of the said algebra with the introduction of a new set of ladder operators gives the so called spectrum generating algebra (SGA).

II. Obtaining a version of deformed Higgs algebra from fusion of pair of such SGAs. As a consequence a 3-boson non-hermitian Hamiltonian with real eigenvalues and orthogonal eigenstates becomes available.

III. Understanding  the non-hermitian Hamiltonian and its eigenstates in the light of  $\partial_{\mathcal { PT}}$-symmetry.

The present article stems from a few of recent studies done by Brody\cite{brody16, brody14}regarding the construction of non-hermitian su(2) generators with real eigenvalues starting from a bi-orthogonal set of vectors. It has been claimed\cite{brody14} that a parallel formalism of quantum mechanics, at least in the framework of finite dimensional Hilbert space, is possible even for relaxing the requirement of hermiticity of the observable. In our case a bi-orthogonal system has been formulated starting from
a pair of orthonormal bases and subsequent construction of a new set of vectors following \cite{baga15} with the help of a suitably chosen transformation. Such a bi-orthogonal system of vectors can be used to construct a set of generators of a Lie algebra which in the present
setting comes out as a parametric deformation of su(2) algebra. The related Jordan-Schwinger
operators $\{J_0^\gamma , J_\pm^\gamma \}$ have been constructed where the operator $J_0^\gamma$ becomes a non-hermitian operator.

It is interesting to note that the complexification of the deformed algebra is not spectrum generating in general for all non-zero values of $\gamma$. However, a $\gamma$-deformed version of complexification is possible at the level of Jordan-Schwinger realization which is spectrum generating in character and resembles the corresponding realization of su(2) for $\gamma = 0$.

Fusion of two such algebras \cite{sunil02} can produce the so called Higgs algebra which is one of the earliest candidates of Polynomial Angular Momentum Algebra(PAMA)\cite{ruan06, higgs79, debergh98}. It has been introduced in \cite{higgs79} to establish the existence of hidden symmetry for coulomb and oscillator potential in a space of constant curvature and also understood as a second order approximation of $su_q (2)$\cite{zhed92}. Here, we have defined a deformed version of Higgs algebra and considered the fusion of $\mathfrak{su}_{\gamma}(2)$ and a deformed version of Dyson-Maleev Realization. To the best of our knowledge such a \textbf{three boson realization of Higgs algebra} is not available in literature.
The related commutation relations result to a 3-boson Hamiltonians involving central elements of the algebras, all of which are non-hermitian operators. A theorem has been proposed to show  the possibility of a \textbf{block-diagonal representation} of the Hamiltonian acting on a homogeneous polynomial space of three indeterminates. The present article considers only two such spaces with degree of homogeneity n=2 and 3. The eigenvalues of the Hamiltonian are found to be real. The  orthogonality of the eigenstates are discussed (in a typical Bargmann-Fock type setting) in view of a \textbf{symmetry induced indefinite inner-product}. 

Since the identification of non-hermitian operators with real spectrum and Bender and Boettcher's\cite{bender98, bender99, bender02} attribution of this possibility to space-time reflection symmetry(${\mathcal {PT}}$-symmetry), so many studies have been undertaken in a variety of contexts relating to discrete symmetries\cite{brody16, mosta02} both in applicational\cite{zheng13, bitt12, rubin07, graefe08, krei16, klai08} and theoretical level\cite{brody17, mosta05}. 

In a recent article by Beygi et. al.\cite{beygi15} the issue of partial ${\mathcal {PT}}$-symmetry has been investigated for N-coupled harmonic oscillator Hamiltonian with purely imaginary coupling term whereas the reality and partial reality of the spectrum are claimed to have direct correspondences with the classical trajectories. The interpretation of all such symmetries can be understood both at the level of Hamiltonian as well as in
its eigenstates\cite{bender02, baga15}. Here we have introduced the notion of partial ${\mathcal {PT}}$-symmetry in relation to boson operators\cite{chakra20, chakra21} which eventually helps us to understand the presence of the same symmetry in the non-hermitian Hamiltonian obtained thus far. 

The notion of $\partial_{\mathcal { PT}}$-symmetry has been introduced through symmetry operators  of different orders and conformity of the Hamiltonian to such symmetry has been understood in terms of commutation relations. It is also observed that in the present situation a choice of eigenvector space is possible where the states can be classified into symmetry conforming, symmetry breaking and symmetry adopting states.

%The indefiniteness of the inner-product space considered here prompts us to investigate the possibility of \textbf{Krein space} structure of the eigenvector space. First, it has been done for spaces corresponding to each of the constituent blocks of the Hamiltonian and finally it is observed that there exists one or more universal symmetry operator for the entire eigenvector space. As the relevant Krein space has finite degree of indefiniteness it also resembles the structure of a \textbf{Pontryagin space}.

\section{Jordan-Schwinger map from bi-orthogonal system}
Given a pair of orthogonal vectors 
$\{\vert u_j\rangle = \frac{1}{\sqrt{2}}
\left(\begin{array}{c}
1 \\
(-1)^{j-1}
\end{array} \right): j = 1,2\}$, $\langle u_j\vert u_k\rangle = \delta_{jk}$ one can construct conventional $\mathfrak{su}(2)$ generators (which are hermitian operators in the usual sense of inner-product in $\mathbf{C^2}$) : 
$\sigma_m = \frac{i^{m+1}}{2}\sum_{j, k=1}^2c^{(m)}_{jk} \vert u_j\rangle \langle u_k\vert : m = 1, 2, 3.$. Here, $c_{jk}^{(1)} = (-1)^j \delta_{jk}$ and $c_{jk}^{(3)} = (-1)^j c_{jk}^{(2)} = (1 -(-1)^j c_{jk}^{(1)})$. The generators follow the commutation : $[\sigma_l , \sigma_m ] = i\sum_{n=1}^3\epsilon_{lmn} \sigma_n $. Now considering boson operators $\{a_j, a^{\dagger}_j : j=1, 2\}$ with $[a_j, a_k^{\dagger}]-\delta_{jk}=[a_j, a_k]=[a_j^{\dagger}, a_k^{\dagger}]=0$, the following Jordan-Schwinger (JS) bi-linear map is obtained as a Lie algebra isomorphism\cite{bieden81}
\begin{equation}
J_m = \frac{i^{m-1}}{2} \sum_{j, k=1}^2c_{jk}^{(4-m)}a_j^\dagger a_k : m = 1, 2, 3.
\end{equation}
Defining $J_{\pm} = J_1 \pm iJ_2$ and using $[a^\dagger_\alpha a_\beta$ , $a^\dagger_\mu a_\nu ] = a^\dagger_\alpha a_\nu \delta_{\beta\mu} -a^\dagger_\mu a_\beta \delta_{\alpha\nu}$ we obtain a spectrum generating algebra(SGA)(taking $J_3 = J_0$ ):
$[ J_0 , J_\pm ]  = \pm J_\pm$
and $[J_+ , J_-] = 2J_0$.

Now, it is possible to extend this procedure in the non-hermitian regime starting from the following definition of bi-orthogonal vectors:
\begin{dfn}
Two pairs of vectors $\{\vert\phi_j\rangle : j=1,2\}$ and $\{\vert\chi_j\rangle : j=1,2\}$ are called bi-orthogonal if $\langle\phi_j\vert\chi_k\rangle=0\forall j\neq k$.
\end{dfn}
Such a bi-orthogonal system can be obtained from straightforward verification of the following theorem.
\begin{theorem}
	Given any pair of vectors  
	$\{\vert v_j\rangle = 
	\left(\begin{array}{c}
	c_j^{(1)} \\
	c_j^{(2)}
	\end{array} \right): j = 1,2\}\in \mathbf{C^2}$ two pairs of vectors $\vert\phi_j\rangle=\omega_0T\vert v_j\rangle$ and $\vert\chi_j\rangle=(T^{-1})^{\dagger}\vert v_j\rangle$ constitute a bi-orthogonal system under the action of a transformation $T = (\cos \frac{\theta}{2} ){\bf 1}_2 +
	2(\cos \frac{\phi}{2}  \sin \frac{\theta}{2} )\sigma_1 - 2(\sin \frac{\phi}{2}  \sin \frac{\theta}{2} )\sigma_2$ with $\omega_0=\cos\theta$ provided $\langle v_j\vert v_k\rangle=(c_j^{(1)})^{\star}c_k^{(1)}+(c_j^{(2)})^{\star}c_k^{(2)}=0\forall j\neq k$.
\end{theorem}

\begin{proof}
	As $T=T^{\dagger}$ with the present sense of inner-product $\langle\phi_j\vert\chi_k\rangle=\langle v_j\vert\omega_0T(T^{-1})^{\dagger}\vert v_k\rangle=\langle v_j\vert\omega_0T(T^{\dagger})^{-1}\vert v_k\rangle=\omega_0\langle v_j\vert v_k\rangle$. Hence the theorem follows from the definition-1.
\end{proof}
 Now, taking $\vert v_j\rangle=\vert u_j\rangle$, $\phi=\pi$ and $\gamma=\sqrt{1-\omega_0^2}=\sin\theta$ a deformed algebra $\mathfrak{su}_{\gamma}(2)$ is obtained with the generators : $\sigma_m^{\gamma} = \frac{i^{m+1}}{2}\sum_{j, k=1}^2\frac{c_{jk}^{(m)}}{\omega_0^{\delta_{m2}}}\vert\phi_j\rangle \langle\chi_k \vert : m = 1, 2, 3$ with the commutator $[ \sigma_l^{\gamma} , \sigma_m^{\gamma}] = i\sum_{n=1}^3\epsilon_{lmn} (1 -\gamma^2 \delta_{n2} )\sigma_n^{\gamma}$. The corresponding JS-map is given by $J_m^{\gamma} = \frac{i^{m-1}}{2}\sum_{j, k=1}^2[c_{jk}^{(4-m)} + i^m (1 - \delta_{m2} )\gamma c_{jk}^{(m)}] a^\dagger_j a_k : m = 1, 2, 3.$.

Introducing a new set of deformed ladder operators with exponent $p$, $p\in \mathbb{R}$ (set of real numbers),
$J_\pm^\gamma = \sum_{r=0}^1(\pm i)^r \omega_0^{r-p} J^\gamma_{r+1}$ and considering $J_3^\gamma = J_0^\gamma$
 the following spectrum generating algebra becomes possible

\begin{equation*}
	[J_0^\gamma, J_\pm^\gamma] = \pm\omega_0(\gamma) J_\pm^\gamma
\end{equation*}
\begin{equation}\label{deformed1}
[J_+^\gamma , J_-^\gamma] = 2\omega_0^{-2p+1}(\gamma)J_0^\gamma.  
\end{equation}
It is to be noted that $J_0^{\gamma}$ is no longer hermitian and $(J^{\gamma}_+)^{\dagger}\neq J^{\gamma}_-$. Casimir of this algebra ${\mathcal {C_J^\gamma}} = \omega_0^{-2} J_0^\gamma (J_0^\gamma \pm \omega_0 ) + \omega_0^{2p-2} J^\gamma_{\mp} J_\pm^\gamma$.

\section{Three-boson realization of Higgs algebra}
\begin{dfn}
A $\gamma$-deformed Higgs algebra $\mathfrak{H^{\gamma}}$ in present sense involving the generators $\{H^{\gamma}_0, H^{\gamma}_{\pm}\}$ is defined by the following commutation relations
\begin{equation}\label{commu1}
[H^{\gamma}_0, H^{\gamma}_{\pm}]=\pm \omega_0 H^{\gamma}_{\pm}
\end{equation}
and
\begin{equation}\label{commu2}
[H^{\gamma}_+, H^{\gamma}_{-}]=4\left[\alpha \omega_0H^{\gamma}_0+\frac{\beta}{\omega_0} (H^{\gamma}_0)^3\right]
\end{equation}
where, $\alpha$ and $\beta$ are functions of central elements of the algebra and $\alpha=\frac{\beta}{8}-\mathcal{H}$. 
\end{dfn}
\begin{remark}
	Infact Higgs algebra has been introduced in the context of non-relativistic Kepler problem in spaces with constant curvature $\pm\beta$ corresponding to hyperboloid and sphere respectively. $\mathcal{H}$ represents the Hamiltonian of the system. It is important to note that all the operators in definition-2 commute with $\mathcal{H}$.
\end{remark}

The function $\alpha$ represents the Hamiltonian of a system and being a central element all the generators commute with it by definition.

\subsection{\textbf{3-boson realization : Fusion of $\mathfrak{su}_{\gamma}(2)$ and a deformed Dyson-Maleev realization.}}

A three boson realization can be obtained from the fusion of $\mathfrak{su}_{\gamma}(2)$ and Dyson-Maleev single boson representation ($\mathfrak{dm}$) involving generators $\{M^{\gamma}_0, M^{\gamma}_{\pm}\}$ by defining the the following operators:
\begin{eqnarray}\label{higgsmaleev}
H^{\gamma}_0=\frac{1}{2}(J_0^{\gamma}-M^{\gamma}_0)\nonumber\\
H^{\gamma}_{\pm}=sJ_\pm^{\gamma}M^{\gamma}_\mp
\end{eqnarray}
and
\begin{equation}
L_0^{\gamma}=\frac{1}{2}(J_0^{\gamma}+M^{\gamma}_0)
\end{equation}
where, $M^{\gamma}_0=\omega_0(c-a_3^{\dagger}a_3)$, $M_+=a_3$ and $M_-=a_3^{\dagger}(2c-a_3^{\dagger}a_3)$ and the commutations $[M_0^{\gamma}, M^{\gamma}_{\pm}]=\pm \omega_0M^{\gamma}_{\pm}$ and $[M_+^{\gamma}, M^{\gamma}_{-}]=\pm \frac{2}{\omega_0}M^{\gamma}_{0}$ hold good along with the Casimir ${\mathcal {C_M^\gamma}} = \omega_0^{-2} M_0^\gamma (M_0^\gamma \pm \omega_0 ) +  M^\gamma_{\mp} M_\pm^\gamma$.  
The operators in equation-\ref{higgsmaleev} follow the commutations in equation-\ref{commu1} and equation-\ref{commu2}
with $p=1$ in the expression of $J_{\pm}^{\gamma}$, ${\mathcal {C_M^\gamma}}={\mathcal {C}}={\mathcal {C_J^\gamma}}$ and $s^2=-\beta$. This makes $\alpha=\frac{\beta}{8}-\mathcal{H}_1$ where, $\mathcal{H}_1\doteq \mathcal{H}_1(a_j, a_j^{\dagger}\mid j=1, 2, 3;\omega_0)=\frac{\beta}{8}+\beta\left(\frac{\mathcal{C}}{\omega_0^4}+\frac{(L_0^{\gamma})^2}{\omega_0^2}\right)$ is a $\textbf{3-boson non-hermitian Hamiltonian}$.

\subsection{\textbf{Spectrum of 3-boson Hamiltonian} ${\mathcal{H}_1\doteq \mathcal{H}_1(a_j, a_j^{\dagger}\mid j=1, 2, 3;\omega_0)}$}
The Hamiltonian has the following expression (taking $\beta=1$)
\begin{equation}
\mathcal{H}_1=\omega_0^{-2}\left(\sum_{i, j=1}^3B_{ij}a_i^{\dagger}a_j+\sum_{i, j, k, l=1}^3V_{ijkl}a_i^{\dagger}a_ja_k^{\dagger}a_l\right)+\rm{const.}
\end{equation}
where, $B_{11}=c\omega_0=-B_{22}, B_{11}=-2\omega_0^2c, B_{12}=C_{21}=ic\gamma\omega_0, B_{13}=B_{31}=B_{23}=B_{32}=0$ and $4V_{1111}=4V_{2222}=-2V_{1122}=1, 4V_{1212}=4V_{2121}=2V_{1221}=-\gamma^2, V_{3333}=\omega_0^2, V_{1112}=V_{1121}=-V_{2212}=-V_{2221}=V_{1211}=-V_{1222}=V_{2111}=-V_{2122}=\frac{i\gamma}{4}, V_{1133}=-V_{2233}=-\omega_0$ and $V_{1233}=V_{2133}=-i\gamma\omega_0$ while all the other $V_{ijkl}$ terms are zero. The Hamiltonian commutes with the number operator $N=a_1^{\dagger}a_1+a_2^{\dagger}a_2+a_3^{\dagger}a_3$ and with the replacements $\{a_j, a_j^{\dagger} : j=1, 2, 3\}\Leftrightarrow \{\zeta_j, \partial_{\zeta_j} : j=1, 2, 3\}$ it can be represented with the help of the following theorem.

\begin{theorem}
$\mathcal{H}_1$ has a block diagonal representation in the basis $\{\vert n_1, n_2, n_3\rangle\Leftrightarrow \zeta_1^{n_1}\zeta_2^{n_2}\zeta_3^{n_3} : n_1+n_2+n_3=n\}$ of the homogeneous polynomial space of degree $n$ with $n_1, n_2, n_3\in \mathbb{Z^+}\cup\{0\}$. For a given $n$, $\mathcal{H}_1=\bigoplus_{j=0}^{n}B^{(n)}_{(j+1)\times (j+1)}$, where $B^{(n)}_{j+1\times j+1}$ is a block diagonal matrix of order $(j+1)\times (j+1)$.
\end{theorem} 

\begin{proof}
	It is obvious to note that for a given $n$ the dimension of the eigenvector space of $\mathcal{H}_1$ is $\frac{1}{2}(n+1)(n+2)$. Now the action of $\mathcal{H}_1$ on a typical basis element $\vert n_1, n_2, n_3\rangle$ can be given by
	
	\begin{eqnarray}\label{block}
	\mathcal{H}_1\vert n_1, n_2, n_3\rangle\nonumber\\
	=\omega_0^{-2}[A_{(0,0,0)}\vert n_1, n_2, n_3\rangle
	+A_{(2,-2,0)}\vert n_1+2, n_2-2, n_3\rangle\nonumber\\
	+A_{(-2,2,0)}\vert n_1-2, n_2+2, n_3\rangle\nonumber\\
	+A_{(1,-1,0)}\vert n_1+1, n_2-1, n_3\rangle
	+A_{(-1,1,0)}\vert n_1-1, n_2+1, n_3\rangle] 
	\end{eqnarray}
	where,	
	\begin{eqnarray}
	A_{(0,0,0)}=\left[\frac{1}{2}(n_1-n_2)+\omega_0(c-n_3)\right]^2-\frac{\gamma^2}{4}(n_1+n_2+2n_1n_2)\nonumber\\
	A_{(2,-2,0)}=A_{(-2,2,0)}=-\frac{\gamma^2}{4}\nonumber\\
	A_{(1,-1,0)}=\frac{i\gamma}{2}n_2(n_1-n_2+1+2\omega_0(c-n_3))\nonumber\\
	A_{(-1,1,0)}=\frac{i\gamma}{2}n_1(n_1-n_2-1+2\omega_0(c-n_3))
	\end{eqnarray}
Now, in view of equation-\ref{block} it is observed that the index $n_3$ is not modified on the right hand side under the action of	$\mathcal{H}_1$. Hence, for a fixed value of $n_3$ the invariant subspace spanned by $\{\zeta_1^{n_1}\zeta_2^{n_2} : n_1, n_2=0\dots n-n_3\}$ has the dimension $n-n_3+1$ representing a block of same dimension. As the largest value of $n_3=n$ the smallest of the block has the dimension $1$, As $n_3$ decreases from $n$ to $0$ with unit interval we get a sequence of blocks with increasing dimension from $1$ to $n+1$.
	\end{proof}

\begin{remark}
	The result is consistent with the fact that the dimension of the eigenvector space $\mathcal{E}^{(n)}$ for a given $n$ is $\sum_{j=0}^{n}(j+1)=\frac{1}{2}(n+1)(n+2)$.
\end{remark}

\subsubsection{Case : 1 $n=2$.}

 A $6\times 6$ representation of the Hamiltonian is possible which can be reduced to a block-diagonal form ($B^{(2)}_{1\times 1}\oplus B^{(2)}_{2\times 2}\oplus B^{(3)}_{3\times 3}$) in the basis $\{\vert 0, 0, 2\rangle; \vert 1, 0, 1\rangle, \vert 0, 1, 1\rangle; \vert 2, 0, 0\rangle, \vert 0, 2, 0\rangle,\vert 1, 1, 0\rangle \}$. Following table can be constructed regarding the eigenvalues and eigenstates (in $\mathcal{E}^{(2)}$) of the Hamiltonian.

\vspace{5mm}
\begin{tabular}{|c|c|} \hline\hline
	\multicolumn{1}{c}{Eigenvalue} &
	\multicolumn{1}{c}{Eigenstate}\\ \hline\hline
	$(c-2)^2$ & $\vert \mathbf{\psi_1}\rangle=\vert 0, 0, 2\rangle $\\ \hline\hline
	$(c-\frac{1}{2})^2$ & $\vert \mathbf{\psi_2}\rangle=\vert 0, 1, 1\rangle - i\left( \frac{1+ \omega_0}{1-\omega_0}\right)^{\frac{1}{2}}\vert 1, 0, 1\rangle $\\ \hline\hline
	$(c-\frac{3}{2})^2$ &  $\vert \mathbf{\psi_3}\rangle=\vert 0, 1, 1\rangle - i\left( \frac{1- \omega_0}{1+\omega_0}\right)^{\frac{1}{2}}\vert 1, 0, 1\rangle $\\ \hline\hline
	0 & $\vert \mathbf{\psi_4}\rangle=\vert 2, 0, 0\rangle -\vert 0, 2, 0\rangle +\frac{2i}{\gamma}\vert 1, 1, 0\rangle$\\ \hline\hline
	1 &  $\vert \mathbf{\psi_5}\rangle=\vert 2, 0, 0\rangle +\vert 0, 2, 0\rangle $\\ \hline\hline
	1 & $\vert \mathbf{\psi_6}\rangle = \vert 2, 0, 0\rangle -\vert 0, 2, 0\rangle +{2i}{\gamma}\vert 1, 1, 0\rangle$\\ \hline\hline
	\end{tabular}

\vspace{5mm}

\begin{remark}
	It is obvious that the eigenvalues are all real and the eigenvalue $"0"$ is doubly degenerate for $c=\frac{1}{2}, \frac{3}{2}, 2$. Eigenvalue $"1"$ which is otherwise doubly degenerate becomes triply degenerate for $c=1, \pm\frac{1}{2},  3, \frac{3}{2}, \frac{5}{2}, 3$. Eigenvalues $"\frac{9}{16}"$, $"\frac{1}{16}"$ and $"\frac{1}{4}"$ are doubly degenerate for $c=\frac{5}{4}$, $c=\frac{7}{4}$ and $c=1$ respectively. It is being observed that for $c>3$ and $c<-\frac{1}{2}$ all the eigenvalues (except $1$) are non-degenerate. The eigenvectors are nontrivially modified i. e.; there is no point of retrieving the eigenstates corresponding to non-deformed case by letting $\gamma\rightarrow 0$. 
\end{remark}
\subsubsection{Case-II : n=3.} A $10\times 10$ representation of the Hamiltonian is possible which is block-diagonal $(B^{(1)}_{4\times 4}\oplus B^{(1)}_{1\times 1}\oplus B^{(1)}_{3\times 3}\oplus B^{(1)}_{2\times 2})$ in the basis $\{\vert 3, 0, 0\rangle, \vert 0, 3, 0\rangle, \vert 2, 1, 0\rangle, \vert 1, 2, 0\rangle, \vert 0, 0, 3\rangle, \vert 2, 0, 1\rangle, \vert 0, 2, 1\rangle, \vert 1, 1, 1\rangle, \vert 1, 0, 2\rangle, \vert 0, 1, 2\rangle \}$. The following  table gives us the eigenvalues and eigenvectors (in $\mathcal{E}^{(3)}$) of the Hamiltonian. 

\vspace{5mm}
\begin{tabular}{|c|c|} \hline\hline
	\multicolumn{1}{c}{Eigenvalue} &
	\multicolumn{1}{c}{Eigenstate}\\ \hline\hline
	$(c-3)^2$ & $\vert \mathbf{\chi_1}\rangle=\vert 0, 0, 3\rangle $\\ \hline\hline
	$(c-\frac{3}{2})^2$ &  $\vert \mathbf{\chi_2}\rangle=\vert 1, 0, 2\rangle +i\left(\frac{1-\omega_0}{1+\omega_0}\right)^{\frac{1}{2}} \vert 0, 1, 2\rangle $\\ \hline\hline
	$(c-\frac{5}{2})^2$ &  $\vert \mathbf{\chi_{3}}\rangle=\vert 1, 0, 2\rangle +i\left(\frac{1+\omega_0}{1-\omega_0}\right)^{\frac{1}{2}} \vert 0, 1, 2\rangle $\\ \hline\hline
	$(c-1)^2$ &  $\vert \mathbf{\chi_4}\rangle=-i\frac{(1-\omega_0^2)^{\frac{1}{2}}}{2}\vert 2, 0, 1\rangle + i\frac{(1-\omega_0^2)^{\frac{1}{2}}}{2}\vert 0, 2, 1\rangle+\vert 1, 1, 1\rangle $\\ \hline\hline
	$c^2$ &  $\vert \mathbf{\chi_5}\rangle=-\frac{i}{2}\left(\frac{1+\omega_0}{1-\omega_0}\right)^{\frac{1}{2}}\vert 2, 0, 1\rangle + \frac{i}{2}\left(\frac{1-\omega_0}{1+\omega_0}\right)^{\frac{1}{2}}\vert 0, 2, 1\rangle+\vert 1, 1, 1\rangle $\\ \hline\hline
	$(c-2)^2$ & $\vert \mathbf{\chi_6}\rangle=-\frac{i}{2}\left(\frac{1-\omega_0}{1+\omega_0}\right)^{\frac{1}{2}}\vert 2, 0, 1\rangle + \frac{i}{2}\left(\frac{1+\omega_0}{1-\omega_0}\right)^{\frac{1}{2}}\vert 0, 2, 1\rangle+\vert 1, 1, 1\rangle $  \\ \hline\hline
	$\frac{1}{4}$ & $\vert \mathbf{\chi_7}\rangle=\frac{3-3\omega_0^2}{3+\omega_0^2}\vert 3, 0, 0\rangle +i\frac{6 (1-\omega_0^2)^{\frac{1}{2}}}{3+\omega_0^2}\vert 3, 0, 0\rangle +\vert 1, 2, 0\rangle$\\ \hline\hline
	$\frac{1}{4}$ &  $\vert \mathbf{\chi_8}\rangle=-i\frac{6 ({1-\omega_0^2})^{\frac{1}{2}}}{3+\omega_0^2}\vert 3, 0, 0\rangle +\frac{3-3\omega_0^2}{3+\omega_0^2}\vert 0, 3, 0\rangle+\vert 2, 1, 0\rangle $\\ \hline\hline
	$\frac{9}{4}$ & $\vert \mathbf{\chi_9}\rangle = \vert 3, 0, 0\rangle +2i(1-\omega_0^2)^{-\frac{1}{2}}\vert 0, 3, 0\rangle +\vert 1, 2, 0\rangle$\\ \hline\hline
	$\frac{9}{4}$ & $\vert \mathbf{\chi_{10}}\rangle=-2i(1-\omega_0^2)^{-\frac{1}{2}}\vert 3, 0, 0\succ+\vert 0, 3, 0\rangle+\vert 2, 1, 0\rangle $\\ \hline\hline

\end{tabular}

\vspace{5mm}

\begin{remark}
	It may be observed that there exist a three fold degeneracy of the eigenvalue $\frac{9}{4}$ for $c=-\frac{3}{2}, \pm\frac{1}{2}, 0,1, \frac{5}{2}, 3, \frac{7}{2}, 4, \frac{9}{2}$ and of the eigenvalue $\frac{1}{4}$ for $c=-\frac{1}{2}, 1, 3, \frac{7}{2}$. Similarly a four fold degeneracy of the eigenvalue $\frac{9}{4}$ is observed for $c=\frac{3}{2}$ and that of $\frac{1}{4}$ is observed for $c=\frac{1}{2}, \frac{3}{2}, 2, \frac{5}{2}$. Here again we observe that for $c>4\frac{1}{2}$ and $c<-\frac{3}{2}$ all eigen values are nondegenerate except $\frac{1}{4}$ and $\frac{9}{4}$.
\end{remark}

\section{\textbf{$\partial_{\mathcal { PT}}$-symmetry of} $\mathcal{H}_1$ \textbf{and its eigenstates}}
Considering boson operators $a_j=2^{-\frac{1}{2}}(x_j+ip_j)$ and $a^{\dagger}_j=2^{-\frac{1}{2}}(x_j-ip_j)$
let us define an operator ${\mathcal P_j}$ whose action has the effect:$x_j\rightarrow -x_j$, $p_j\rightarrow -p_j$ and an operator ${\mathcal T}$ whose action has the effect : $x_j\rightarrow x_j$, $p_j\rightarrow -p_j$ and $i\rightarrow -i$. It is crucial to note that ${\mathcal P_j}$ can only act on $j$-th coordinate and momentum and therefore indifferent to $\{x_k, p_k \forall k\neq j\}$ whereas the action of ${\mathcal T}$ is independent of $j$. ${\mathcal P_j}$ can be called the $j$-th
partial parity operator and ${\mathcal T}$ is  called the time-reversal operator\cite{bender02, beygi15}. This means  
$\{{\mathcal P_j} , a_j \}_+ = 0 = \{{\mathcal P_j} , a^\dagger_j \}_+$
and $[{\mathcal P_j} , a_k ] = 0 = [{\mathcal P_j} , a^\dagger_k ]$ for
$k \neq j$.The time reversal operator
${\mathcal T}$ follows the rule $[{\mathcal T} , a_k ] = 0 =
[{\mathcal T} , a^\dagger_k ]$. 
\begin{dfn}\label{partial}
	Let us consider a multiboson operator $\Upsilon\doteq \Upsilon(a_j ;  a_j^{\dagger}: j=1\dots n, i)$. The action of $j$-th $\partial_{\mathcal { PT}}$ operator $\Pi^{(j)}_T=\mathcal{P}_j\mathcal{T}$ on $\Upsilon$ is a composite action given by the following transformation : 
	\begin{eqnarray}
	\Pi^{(j)}_T : \Upsilon(a_1, \dots, a_j,\dots ,a_n;  a_1^{\dagger}, \dots,a_j^{\dagger},\dots, a_n^{\dagger}; i)\nonumber\\
	\rightarrow\Upsilon(a_1, \dots, -a_j,\dots ,a_n;  a_1^{\dagger}, \dots,-a_j^{\dagger},\dots, a_n^{\dagger}; -i).
	\end{eqnarray}
\end{dfn}
\begin{remark}
	The definition can be extended by introducing the operator $\Pi^{(j_1j_2\dots j_l)}_T=\mathcal{P}_{j_1}\dots \mathcal{P}_{j_l}\mathcal{T}$ which implies simutaneous changes performed on corresponding boson operators. It is to be mentioned that if any operator $\Upsilon$ remains unchanged under such action it is claimed to have partial $\partial_{\mathcal { PT}}$ -symmetry  in the corresponding boson-operators. In short this may be written as $[\Pi^{(j_1j_2\dots j_l)}_T, \Upsilon]=0$ A global ${\mathcal {PT}}$ -symmetry operator $\Pi_T$ can therefore be understood by the operation : $\Pi_T\doteq \Pi^{(1, \dots, n)}_T : \Upsilon(a_1, \dots, a_j,\dots ,a_n;  a_1^{\dagger}, \dots,a_j^{\dagger},\dots, a_n^{\dagger}; i)
	\rightarrow\Upsilon(-a_1, \dots, -a_j,\dots ,-a_n;  -a_1^{\dagger}, \dots,-a_j^{\dagger},\dots, -a_n^{\dagger}; -i)$. $\Upsilon$ can be claimed to have $l$-th order $\partial_{\mathcal { PT}}$ symmetry in the corresponding multi-index $l$ if $[\Pi^{(j_1, \dots, j_l)}_T, \Upsilon]=0$.
\end{remark}

\subsection{\textbf{$\partial_{\mathcal { PT}}$-symmetry} in Fock Space}
	Considering the Bargmann-Fock correspondence : $a_j^{\dagger} = \zeta_j$ and $ a_j = \partial_{\zeta_j}$, $\{\zeta_j : j=1, 2, 3\}$ being complex variables, the operator $\mathcal{H}_1$
	 takes the form
	\begin{eqnarray}
	\mathcal{H}_1=\omega_0^{-2}\left(\sum_{i, j=1}^3B_{ij}\zeta_i\partial_{\zeta_j}+\sum_{i, j, k, l=1}^3V_{ijkl}\zeta_i\partial_{\zeta_j}\zeta_k\partial_{\zeta_l}\right)+\rm{const.}
	\end{eqnarray}
	
	The eigenfunctions of the operator $\mathcal{H}_1$ can be sought in Fock (or Segal-Bargmann) space $(\mathcal{F}^2(\mathbb{C}^3))$ which is a separable complex Hilbert space of entire functions (of the complex variables $\zeta_1$, $\zeta_2$ and $\zeta_3$) equipped with an inner-product 
	\begin{eqnarray}
	\langle \psi, \phi\rangle=\int_{W(\zeta_1)}\int_{W(\zeta_2)}\int_{W(\zeta_3)} \psi(\zeta_1,\zeta_2, \zeta_3)\overline{\phi(\zeta_1, \zeta_2, \zeta_3)}\nonumber\\
	\;\; {\rm{with}} \int_{W(\zeta_1)}\int_{W(\zeta_2)}\int_{W(\zeta_3)}\equiv \int\int\int d{W(\zeta_1)}d{W(\zeta_2)}d{W(\zeta_3)}.
	\end{eqnarray} 
	Here, ${d}{W(u)}=\frac{1}{\pi}e^{-\vert u\vert^2}{d}({\rm{Re}}(u)) d({\rm{Im}}(u))$ represents the relevant Gaussian measure relative to the complex variable $u$.
\begin{dfn}
	In $(\mathcal{F}^2(\mathbb{C}^1))$  the \textbf{weighted composition conjugation} $\mathcal{C}_{(\vartheta, \eta, \upsilon)}\psi(\zeta)$  \cite{hai18,hai16,hai18c} is defined by the following action
	\begin{eqnarray}
	\mathcal{C}_{(\vartheta, \eta, \upsilon)}\psi(\zeta)=\upsilon e^{\eta\zeta}\overline{\psi(\overline{\vartheta\zeta+\eta})}.
	\end{eqnarray}
\end{dfn}

	Here, $\zeta$ is a complex variable and $\{\vartheta, \eta, \upsilon\}$ are complex numbers satisfying the set of necessary and sufficient conditions : $\vert\vartheta\vert=1,\bar{\vartheta}\eta+\bar{\eta}=0$ and $\vert\upsilon\vert^2e^{\vert\eta\vert^2}=1$. 
	
	\begin{remark}
		The anti-linear operator $\mathcal{C}_{(\vartheta, \eta, \upsilon)}$ is a \textbf{conjugation} since it is involutive and isometric \cite{hai16}.
	\end{remark}
	 The action of the operator $\mathcal{PT}$ is equivalent to the choice : $\vartheta=-1=-\upsilon,\eta=0$ which results to the following equation
	\begin{eqnarray}
	\mathcal{C}_{(\vartheta, 0, 1)}\vert_{\vartheta=-1}\psi(\zeta)=\overline{\psi(\overline{-\zeta})}.
	\end{eqnarray}
	Similarly, the action of $\mathcal{T}$ is indicative of the choice : $\vartheta=1, \eta=0, \upsilon=1$ giving
	\begin{eqnarray}
	\mathcal{C}_{(\vartheta, 0, 1)}\vert_{\vartheta=1}\psi(\zeta)=\overline{\psi(\overline{\zeta})}.
	\end{eqnarray}
	If $\psi$ is a function of several complex variables $\{\zeta_j : j=1\dots n\}$ one can define an operator $\mathcal{C}_{(\vartheta_j,\eta_j, \upsilon_j : j=1\dots n )}$ with the action
	\begin{eqnarray}
	\mathcal{C}_{(\vartheta_j,\eta_j= 0,\upsilon_j= 1 : j=1\dots n)}\psi(\zeta_1,\dots,\zeta_j,\dots,\zeta_n)=\overline{\psi(\overline{\vartheta_1\zeta_1},\dots,\overline{\vartheta_j\zeta_j},\dots, \overline{\vartheta_n\zeta_n})}
	\end{eqnarray} 
	Let us introduce an operator $\mathcal{C}^{(j)}_n=\mathcal{C}_{(\vartheta_j,\eta_j= 0,\upsilon_j= 1 ; j=1\dots n)}\vert_{\vartheta_1=1,\dots,\vartheta_j=-1,\dots,\vartheta_n=1}$ as $j$-th partial $\mathcal{PT}$ symmetry operator through the following action
	\begin{eqnarray}
	\mathcal{C}^{(j)}_n\psi(\zeta_1,\dots,\zeta_j,\dots,\zeta_n)=\overline{\psi(\bar{\zeta_1},\dots,\overline{-\zeta_j},\dots,\bar{\zeta_n})}.
	\end{eqnarray}
Depending upon each index there can be $n$ such operators for a typical $n$-variable situation. Let us call this set $\{\mathcal{C}^{(j)}_n : j=1\cdots n\}$ as the set of 1st order partial $\mathcal{PT}$ symmetry operators. Similarly, $\{\mathcal{C}^{(jk)}_n : j,k=1\cdots n\; j\neq k\}$ represents the set of $\frac{1}{2}n(n-1)$ number of 2nd order symmetry operators. For an $n$-variable system there would be $(^n_m)$ number of $m$-th order ($m\leq n$) symmetry operators for obvious reason.
 
	The global $\mathcal{PT}$ symmetry operator $\mathcal{C}_n^{(1\cdots n)}$ (which will be written as $\mathcal{C}_n$) can therefore be defined through the action
	\begin{eqnarray}
	\mathcal{C}_n\psi(\zeta_1,\dots,\zeta_j,\dots,\zeta_n)=\overline{\psi(\overline{-\zeta_1},\dots,\overline{-\zeta_j},\dots,\overline{-\zeta_n})}.
	\end{eqnarray}
	For our present purpose we shall only consider the operators $\mathcal{C}_3$ , $\{\mathcal{C}_3^{(j)} : j=1, 2, 3\}$ and $\{\mathcal{C}_3^{(jk)} : j, k=1, 2, 3 ; j\neq k\}$. Now, global and partial  $\mathcal{PT}$ symmetries of any
	function $\psi(\zeta_1, \zeta_2, \zeta_3)$ are understood through the following equations
	\begin{eqnarray}
	\mathcal{C}_3\psi(\zeta_1, \zeta_2, \zeta_3)=\psi(\zeta_1, \zeta_2, \zeta_3)\:\:{\rm and}\:\:\mathcal{C}_3^{(j)}\psi(\zeta_1, \zeta_2, \zeta_3)=\psi(\zeta_1, \zeta_2, \zeta_3)\nonumber\\
	\:\:\forall\:\: j=1, 2, 3\nonumber\\
	\mathcal{C}_3^{(jk)}\psi(\zeta_1, \zeta_2, \zeta_3)=\psi(\zeta_1, \zeta_2, \zeta_3)\nonumber\\
	\:\:\forall\:\: j, k=1, 2, 3 ; j\neq k\nonumber\\
	\end{eqnarray}
	respectively.
	The global and partial  $\mathcal{PT}$ symmetries of any operator $\mathcal{Z}$ are equivalent to the following equalities
	\begin{eqnarray}\label{conju1}
	\mathcal{C}_3\mathcal{Z}\mathcal{C}_3=\mathcal{Z}\:\:{\rm and}\:\:\mathcal{C}^{(j)}_3\mathcal{Z}\mathcal{C}^{(j)}_3=\mathcal{Z}\:\:\forall\:\:j=1, 2, 3\nonumber\\
	\mathcal{C}^{(jk)}_3\mathcal{Z}\mathcal{C}^{(jk)}_3=\mathcal{Z}\:\:\forall\:\:j,k=1, 2, 3 \:\:{\rm with}\:\: j\neq k
	\end{eqnarray}
	respectively. 
	With this definition of conjugation we shall verify a number of properties of the Hamiltonian $\mathcal{H}_1$ through the notion of \textbf{Reproducing Kernel Hilbert Space}.
	
	\begin{dfn}
		A function of the form $K^{[m_1, m_2, m_3]}_{\zeta_1, \zeta_2, \zeta_3}(u_1, u_2, u_3)=\prod_{j=1}^3 u_j^{m_j}e^{u_j\overline{\zeta_j}}$($m_1, m_2, m_3\in\mathbf{N}$ and $\zeta_j,u_j\in \mathbb{C}\:\:\forall\:\:j=1, 2, 3$) is called a kernel function (or a \textbf{reproducing kernel}) which satisfies the condition 
		\begin{eqnarray}
		\psi^{(m_1,m_2,m_3)}(\zeta_1, \zeta_2, \zeta_3)=\left\langle\psi, K^{[m_1, m_2, m_3]}_{\zeta_1, \zeta_2, \zeta_3}\right\rangle\nonumber\\
		=\int_{W(u_1)}\int_{W(u_2)}\int_{W(u_3)}\psi(u_1, u_2, u_3)\overline{K^{[m_1, m_2, m_3]}_{\zeta_1, \zeta_2, \zeta_3}} 
		\end{eqnarray}
	\end{dfn}
	
Here, $\psi\in\mathcal{F}^2(\mathbb{C}^3)$, $\psi^{(m_1, m_2, m_3)}(\zeta_1, \zeta_2, \zeta_3)=\partial_{\zeta_1}^{m_1}\partial_{\zeta_2}^{m_2}\partial_{\zeta_3}^{m_3}\psi$ and $\psi^{(0, 0, 0)}(\zeta_1, \zeta_2, \zeta_3)\equiv\psi(\zeta_1, \zeta_2, \zeta_3)$. Such a kernel function renders the present Fock space ($\mathcal{F}^2(\mathbb{C}^3)$) meaningful in the sense of so called \textbf{Reproducing Kernel Hilbert Space}. The existence of such a kernel function facilitates our way of demonstrating various properties of functions and operators relative to the Hilbert space under consideration.

\begin{prop}
	$\mathcal{H}_1^{\star}\neq \mathcal{H}_1$. Where, $\mathcal{H}_1^{\star}$ the adjoint of the operator $\mathcal{H}_1$ and the adjoint is defined as
	$\left\langle\mathcal{H}_{1}\psi,K^{[m_1, m_2, m_3]}_{\zeta_1, \zeta_2, \zeta_3}\right\rangle=\left\langle\psi,\mathcal{H}_{1}^{\star}K^{[m_1, m_2, m_3]}_{\zeta_1, \zeta_2, \zeta_3}\right\rangle$.
\end{prop}	
\begin{proof}
	We shall first show that 
	
	$\left\langle u_1\partial_{u_1}\psi(u_1,u_2,u_3),K^{[m_1, m_2, m_3]}_{\zeta_1, \zeta_2, \zeta_3}\right\rangle=\left\langle \psi(u_1,u_2,u_3),u_1\partial_{u_1}K^{[m_1, m_2, m_3]}_{\zeta_1, \zeta_2, \zeta_3}\right\rangle$.
	
	\begin{eqnarray}
	\left\langle u_1\partial_{u_1}\psi(u_1,u_2,u_3),K^{[m_1, m_2, m_3]}_{\zeta_1, \zeta_2, \zeta_3}\right\rangle\nonumber\\
	=\left\langle u_1\partial_{u_1}\psi{(u_1, u_2,u_3)},\prod_{j=1}^3u_j^{m_j}e^{u_j\overline{\zeta_j}}\right\rangle\nonumber\\
	=\int_{W(u_1)}\int_{W(u_2)}\int_{W(u_3)}u_1\partial_{u_1}\psi(u_1, u_2, u_3)\prod_{j=1}^3\overline{u_j}^{m_j}e^{\overline{u_j}{\zeta_j}}\nonumber\\
	=\int_{W(u_1)}\int_{W(u_2)}\int_{W(u_3)}u_1\left[\int_{W(v_1)}\int_{W(v_2)}\int_{W(v_3)}\overline{v_1}\psi(v_1, v_2, v_3)\prod_{j=1}^3e^{\overline{v_j}{u_j}}\right]\prod_{j=1}^3\overline{u_j}^{m_j}e^{\overline{u_j}{\zeta_j}}\nonumber\\
	=\int_{W(v_1)}\int_{W(v_2)}\int_{W(v_3)}\overline{v_1}\psi(v_1, v_2,v_3)\left[\int_{W(u_1)}\int_{W(u_2)}\int_{W(u_3)}u_1\prod_{j=1}^3\overline{u_j}^{m_j}e^{\overline{u_j}{\zeta_j}}e^{\overline{v_j}{u_j}}\right]\nonumber\\
	=\int_{W(v_1)}\int_{W(v_2)}\int_{W(v_3)}\overline{v_1}\psi(v_1, v_2, v_3)\left\langle u_1\prod_{j=1}^3e^{\overline{v_j}{u_j}}, \prod_{j=1}^3{u_j}^{m_j}e^{{u_j}\overline{\zeta_j}}\right\rangle\nonumber\\
	=\int_{W(v_1)}\int_{W(v_2)}\int_{W(v_3)}\overline{v_1}\psi(v_1, v_2, v_3)\partial_{\zeta_1}^{m_1}\partial_{\zeta_2}^{m_2}\partial_{\zeta_3}^{m_3}\left(\zeta_1\prod_{j=1}^3e^{\overline{v_j}{\zeta_j}}\right)\nonumber\\
	\int_{W(v_1)}\int_{W(v_2)}\int_{W(v_3)}\overline{v_1}\psi(v_1, v_2,v_3)\partial_{\zeta_1}^{m_1}\partial_{\zeta_2}^{m_2}\partial_{\zeta_3}^{m_3}\partial_{\overline{v_1}}\prod_{j=1}^3(e^{\overline{v_j}{\zeta_j}})\nonumber\\
	=\int_{W(v_1)}\int_{W(v_2)}\int_{W(v_3)}\psi(v_1, v_2,v_3)\overline{v_1}\partial_{\overline{v_1}}\partial_{\zeta_1}^{m_1}\partial_{\zeta_2}^{m_2}\partial_{\zeta_3}^{m_3}\prod_{j=1}^3(e^{\overline{v_j}{\zeta_j}})\nonumber\\
	=\left\langle\psi(v_1, v_2, v_3), v_1\partial_{v_1}K^{[m_1, m_2, m_3]}_{\zeta_1, \zeta_2, \zeta_3}(v_1, v_2, v_3)\right\rangle.
	\end{eqnarray}
\end{proof}
An identical argument holds for $u_j\partial_{u_j}$ for $j=2, 3$ whereas, $iu_j\partial_{u_k}=-iu_j\partial_{u_k}$ for $j\neq k$. Using this fact in the expression of the hamiltonian the proposition can be verified. 

\begin{prop}
	$\mathcal{H}_1$ is $\mathcal{C}_3$-self-adjoint i. e.; $\mathcal{C}_3\mathcal{H}_{1}^{\star}\mathcal{C}_3=\mathcal{H}_{1}$. 
\end{prop}	
\begin{proof}
	We shall verify the fact
	
	 $\mathcal{C}_3(u_1\partial_{u_1})^{\star}\mathcal{C}_3K^{[m_1, m_2, m_3]}_{\zeta_1, \zeta_2, \zeta_3}(u_1, u_2, u_3)=(u_1\partial_{u_1})K^{[m_1, m_2, m_3]}_{\zeta_1, \zeta_2, \zeta_3}(u_1, u_2, u_3)$.
	 
	 Now, 
	 \begin{eqnarray}
	 \mathcal{C}_3(u_1\partial_{u_1})^{\star}\mathcal{C}_3K^{[m_1, m_2, m_3]}_{\zeta_1, \zeta_2, \zeta_3}(u_1, u_2, u_3)\nonumber\\
	 =\mathcal{C}_{3}(u_1\partial_{u_1})^{\star}\mathcal{C}_{3}\prod_{j=1}^3u_j^{m_j}e^{u_j\overline{\zeta_j}}\nonumber\\
	 =\mathcal{C}_{3}u_1\partial_{u_1}\prod_{j=1}^3\overline{\overline{(-u_j)}^{m_1}}}\overline{e^{\overline{-u_j}\overline{\zeta_j}}\nonumber\\
	 =\mathcal{C}_{3}u_1\partial_{u_1}(-1)^{m_1+m_2+m_3}\prod_{j=1}^3u_j^{m_j}e^{-u_j{\zeta_j}}\nonumber\\
	 =\mathcal{C}_{3}u_1(-1)^{m_1+m_2+m_3}[(-{\zeta_1})+m_1u_1^{-1}]\prod_{j=1}^3u_j^{m_j}e^{-u_j{\zeta_j}}\nonumber\\
	 =-u_1(-1)^{m_1+m_2+m_3}(-1)^{m_1+m_2+m_3}[(-\overline{\zeta_1})\nonumber\\
	 +m_1(-u_1)^{-1}]\prod_{j=1}^3e^{u_j\overline{\zeta_j}}\nonumber\\
	 =u_1\partial_{u_1}K^{[m_1, m_2, m_3]}_{\zeta_1, \zeta_2, \zeta_3}(u_1, u_2, u_3).
	 \end{eqnarray}
\end{proof}	
Similar result can be verified 	$u_j\partial_{u_j}$ for $j=2, 3$ and $iu_j\partial_{u_k}=iu_j\partial_{u_k}$ for $j\neq k$. Using all these results in the expression of $\mathcal{H}_1$, $\mathcal{C}_3$-selfadjointness can be verified.

The following prposition verifies the existence of various order of partial $\mathcal{PT}$-symmetry of the Hamiltonian.
\begin{prop}
	$\mathcal{C}_3^{(j)}\mathcal{H}_1\mathcal{C}_3^{(j)}=\mathcal{H}_1$, for j=1,2,3  $\mathcal{C}_3^{(jk)}\mathcal{H}_1\mathcal{C}_3^{(jk)}=\mathcal{H}_1$ for $j=1, 2$ and
	$k=3$ 
	$\mathcal{C}_3^{(12)}\mathcal{H}_1\mathcal{C}_3^{(12)}\neq\mathcal{H}_1$
	$\mathcal{C}_3^{(123)}\mathcal{H}_1\mathcal{C}_3^{(123)}\neq\mathcal{H}_1$  
\end{prop}
\begin{proof}
	We shall only verify the that
	
	 $\mathcal{C}_3^{(1)}u_1\partial_{u_1}\mathcal{C}_3^{(j)}K^{[m_1, m_2,m_3]}_{\zeta_1, \zeta_2,\zeta_3}(u_1, u_2, u_3)=u_1\partial_{u_1}K^{[m_1, m_2, m_3]}_{\zeta_1, \zeta_2, u_3}(u_1, u_2, u_3)$ and
	 
	  $\mathcal{C}_3^{(1)}iu_1\partial_{u_2}\mathcal{C}_3^{(j)}K^{[m_1, m_2,m_3]}_{\zeta_1, \zeta_2,\zeta_3}(u_1, u_2, u_3)=iu_1\partial_{u_2}K^{[m_1, m_2, m_3]}_{\zeta_1, \zeta_2, u_3}(u_1, u_2, u_3)$through the following steps

\begin{eqnarray}
\mathcal{C}_3^{(1)}u_1\partial_{u_1}\mathcal{C}_3^{(1)}K^{[m_1, m_2, m_3]}_{\zeta_1, \zeta_2, \zeta_3}(u_1, u_2, u_3)\nonumber\\
=\mathcal{C}_3^{(1)}u_1\partial_{u_1}[(-u_1)^{m_1}e^{-u_1\zeta_1}]\prod_{j\neq 1}^3u_j^{m_j}e^{u_j\zeta_j}\nonumber\\
=\mathcal{C}_3^{(1)}u_1(-1)^{m_1}[m_1u_1^{-1}+(-\zeta_1)]u_1^{m_1}e^{-u_1\zeta_1}\prod_{j\neq 1}^3u_j^{m_j}e^{u_j\zeta_j}\nonumber\\
=u_1\partial_{u_1}K^{[m_1, m_2, m_3]}_{\zeta_1, \zeta_2, u_3}(u_1, u_2, u_3)
\end{eqnarray}
and 
\begin{eqnarray}
\mathcal{C}_3^{(1)}iu_1\partial_{u_2}\mathcal{C}_3^{(1)}K^{[m_1, m_2, m_3]}_{\zeta_1, \zeta_2, \zeta_3}(u_1, u_2, u_3)\nonumber\\
=\mathcal{C}_3^{(1)}iu_1\partial_{u_2}[(-u_1)^{m_1}e^{-u_1\zeta_1}]\prod_{j\neq 1}^3(u_j)^{m_j}e^{u_j\zeta_j}\nonumber\\
=\mathcal{C}_3^{(1)}iu_1(-1)^{m_1}[(u_1)^{m_1}e^{-u_1\zeta_1}][m_2u_2^{-1}+\zeta_2]\prod_{j\neq 1}^3(u_j)^{m_j}e^{u_j\zeta_j}\nonumber\\
=iu_1\partial_{u_2}K^{[m_1, m_2, m_3]}_{\zeta_1, \zeta_2, u_3}(u_1, u_2, u_3)
\end{eqnarray}	
Using these and similar results in the expression of the Hamiltonian we can verify part-1 of the proposition-1. The other parts can be proved in similar fashion.
\end{proof}

The eigenvector space can be configured in terms of $\partial_{\mathcal { PT}}$-symmetric states, $\partial_{\mathcal { PT}}$-symmetry breaking states and $\partial_{\mathcal { PT}}$-adopting states corresponding to the actions of different types of symmetry operators.

For $n=2$

\vspace{5mm}
\begin{tabular}{|c|c|c|} \hline\hline\hline
	\multicolumn{1}{p{2cm}}{$\partial_{\mathcal { PT}}$-Symmetry Operator} &
	\multicolumn{1}{p{3cm}}{$\partial_{\mathcal { PT}}$-Symmetric State} &
	\multicolumn{1}{p{3cm}}{$\partial_{\mathcal { PT}}$-Symmetry Breaking State}\\ 
	 \hline\hline\hline
	 $\mathcal{C}_3^{(1)}$ & $\vert\mathbf{\psi_1}\rangle, \vert\mathbf{\psi_2}\rangle, \vert\mathbf{\psi_3}\rangle, \vert\mathbf{\psi_4}\rangle, \vert\mathbf{\psi_5}\rangle, \vert\mathbf{\psi_6}\rangle$ & \textrm{none}\\ \hline\hline\hline
	 $\mathcal{C}_3^{(2)}$ & $\vert\mathbf{\psi_1}\rangle, \vert\mathbf{\psi_4}\rangle, \vert\mathbf{\psi_5}\rangle, \vert\mathbf{\psi_6}\rangle$ & $\vert\mathbf{\psi_2}\rangle, \vert\mathbf{\psi_3}\rangle$ \\ \hline\hline\hline
	 $\mathcal{C}_3^{(3)}$ & $\vert\mathbf{\psi_1}\rangle, \vert\mathbf{\psi_4}\rangle, \vert\mathbf{\psi_5}\rangle, \vert\mathbf{\psi_6}\rangle$ & $\vert\mathbf{\psi_2}\rangle, \vert\mathbf{\psi_3}\rangle$\\ \hline\hline\hline
	 $\mathcal{C}_3^{(13)}$ & $\vert\mathbf{\psi_1}\rangle, \vert\mathbf{\psi_4}\rangle, \vert\mathbf{\psi_5}\rangle, \vert\mathbf{\psi_6}\rangle$ & $\vert\mathbf{\psi_2}\rangle, \vert\mathbf{\psi_3}\rangle$\\ \hline\hline\hline
	 $\mathcal{C}_3^{(23)}$ & $\vert\mathbf{\psi_1}\rangle, \vert\mathbf{\psi_2}\rangle, \vert\mathbf{\psi_3}\rangle, \vert\mathbf{\psi_4}\rangle, \vert\mathbf{\psi_5}\rangle, \vert\mathbf{\psi_6}\rangle$ & \textrm{none} \\ \hline\hline\hline
	 \end{tabular}

\vspace{5mm}

\begin{remark}
It is interesting to note that the present Hamiltonian $\mathcal{H}_1$ is not $\partial_{\mathcal { PT}}$-symmetric under the action of $\mathcal{C}_3^{(12)}$ though the states $\vert\mathbf{\psi_1}\rangle, \vert\mathbf{\psi_5}\rangle$ are found to be $\partial_{\mathcal { PT}}$-symmetric while $\vert\mathbf{\psi_2}\rangle, \vert\mathbf{\psi_3}\rangle,   \vert\mathbf{\psi_4}\rangle, \vert\mathbf{\psi_6}\rangle$ are not. One may call the states $\vert\mathbf{\psi_1}\rangle, \vert\mathbf{\psi_5}\rangle$ as $\partial_{\mathcal { PT}}-\textbf{symmetry adopting}$ states. Exactly same observation holds for the operator $\mathcal{C}_3^{(123)}$ which is also a global ${\mathcal { PT}}$ operator for the system. Here again $\mathcal{C}_3^{(123)} \mathcal{H}_1\mathcal{C}_3^{(123)}\neq \mathcal{H}_1 $ still the states $\vert\mathbf{\psi_1}\rangle, \vert\mathbf{\psi_5}\rangle$ are symmetric under its action. .  
\end{remark}

For $n=3$ the following observation can be made

\vspace{5mm}
\begin{tabular}{|c|p{6cm}|p{6cm}|} \hline\hline\hline
	\multicolumn{1}{p{2cm}}{$\partial_{\mathcal { PT}}$-Symmetry Operator} &
	\multicolumn{1}{p{3cm}}{$\partial_{\mathcal { PT}}$-Symmetric State} &
	\multicolumn{1}{p{3cm}}{$\partial_{\mathcal { PT}}$-Symmetry Breaking State}\\ 
	\hline\hline\hline
	$\mathcal{C}_3^{(1)}$ & $\vert\mathbf{\chi_1}\rangle, \vert\mathbf{\chi_8}\rangle, \vert\mathbf{\chi_{10}}\rangle $ & $\vert\mathbf{\chi_2}\rangle, \vert\mathbf{\chi_{3}}\rangle, \vert\chi_4\rangle, \vert\mathbf{\chi_5}\rangle, \vert\mathbf{\chi_6}\rangle, \vert\chi_7\rangle, \vert\chi_9\rangle $\\ \hline\hline\hline
	$\mathcal{C}_3^{(2)}$ & $\vert\mathbf{\chi_1}\rangle, \vert\mathbf{\chi_2}\rangle, \vert\mathbf{\chi_{3}}\rangle,\vert\mathbf{\chi_7}\rangle, \vert\mathbf{\chi_9}\rangle$ & $\vert\mathbf{\chi_4}\rangle, \vert\mathbf{\chi_5}\rangle, \vert\mathbf{\chi_6}\rangle, \vert\mathbf{\chi_8}\rangle, \vert\mathbf{\chi_{10}}\rangle$\\ \hline\hline\hline
	$\mathcal{C}_3^{(3)}$ & $ \vert\mathbf{\chi_7} \rangle,\vert\mathbf{\chi_8}\rangle, \vert\mathbf{\chi_9}\rangle,\vert\mathbf{\chi_{10}}\rangle $ & $\vert\mathbf{\chi_1}\rangle,\vert\mathbf{\chi_2}\rangle, \vert\mathbf{\chi_{3}}\rangle,\vert\mathbf{\chi_4}\rangle, \vert\mathbf{\chi_5}\rangle,\vert\mathbf{\chi_6}\rangle$\\ \hline\hline\hline
	$\mathcal{C}_3^{(13)}$ & $ \vert\mathbf{\chi_4}\rangle,\vert\mathbf{\chi_5}\rangle, \vert\mathbf{\chi_6}\rangle,\vert\mathbf{\chi_8}\rangle, \vert\mathbf{\chi_{10}}\rangle $ & $\vert\mathbf{\chi_1}\rangle,\vert\mathbf{\chi_2}\rangle, \vert\mathbf{\chi_{3}}\rangle,\vert\mathbf{\chi_7}\rangle, \vert\mathbf{\chi_9}\rangle$\\ \hline\hline\hline
	$\mathcal{C}_3^{(23)}$ & $\vert\mathbf{\chi_2}\rangle, \vert\mathbf{\chi_{3}}\rangle, \vert\mathbf{\chi_4}\rangle,\vert\mathbf{\chi_5}\rangle, \vert\mathbf{\chi_6}\rangle,\vert\mathbf{\chi_7}\rangle, \vert\mathbf{\chi_9}\rangle $ & $\vert\mathbf{\chi_1}\rangle, \vert\mathbf{\chi_8}\rangle, \vert\mathbf{\chi_{10}}\rangle$  \\ \hline\hline\hline
\end{tabular}

\vspace{5mm} 

\begin{remark}
	$\mathcal{C}_3^{(12)}$ and $\mathcal{C}_3^{(123)}$ do not act as $\partial_{\mathcal { PT}}$ symmetry operators for the system. 
\end{remark}

\section{Comments}
The main objective of our present discussion is to construct boson realization of non-hermitian operators related to deformed linear and nonlinear algebras and a Hamiltonian obtained thereof from their fusion. One of the deformed algebras has been obtained with the aid of a bi-orthogonal system of vectors. Fusion of two such algebras leads to Higgs algebra of cubic polynomial type. The Hamiltonian of the system admitting Higgs algebra is found to be non-hermitian still admitting real eigenvalues. The notion of partial ${\mathcal {PT}}$-symmetry has been investigated for the Hamiltonian and its eigen states. 

%Finally the orthogonality of eigen states has been discussed considering a symmetry  induced inner-product space which is found to be indefinite in nature. The Krein space structure of such an inner-product space is also investigated. It is to be mentioned that several space and algebraic features of this kind of many body systems are not yet explored in non-hermitian domain and a number unusual features may crop up in any stage of its development especially in the context of partial ${\mathcal { PT}}$ symmetry. 


\section*{References}
\begin{thebibliography}{99}
\bibitem{brody16}D. C. Brody, J. Phys. A: Math. Theor. $\bf{49}$ (2016) 10LT03.
\bibitem{brody14}D. C. Brody, J. Phys. A: Math. Theor. $\bf{47}$ (2014) 035305.	
\bibitem{baga15}Editors: F. Bagarello, J. P. Gazeau, F. H. Szafraniec, M. Znojil:{\it Non-Selfadjoint Operators in
	Quantum Physics} (chapter:3 page no. 123-124 and 277-278 regarding basis construction), (chapter:6, specifically page no. 323-324 regarding ${\mathcal {PT}}$ symmetry ), Wiley(2015).
\bibitem{sunil02}V. Sunil Kumar, B. A. Bambah and R. Jagannathan, Mod. Phys. Lett. $\bf{17}$ (2002) No. 24, 1559.
\bibitem{ruan06}D. Ruan, J. Math. Chem. $\bf{39}$ (2006) No. 2, 417.
\bibitem{higgs79}P. W. Higgs, J. Phys. A: Math. Gen. $\bf{12}$ (1979) No. 3, 309.
\bibitem{debergh98}N. Debergh, J. Phys. A: Math. Gen. $\bf{31}$ (1998) 4013.
\bibitem{zhed92}A. S. Zhedanov, Mod. Phys. Lett. A $\bf{82}$ (1992) No. 6, 507.
\bibitem{bender98}C. M. Bender and S. Boettcher , Phys. Rev. Lett. $\bf{89}$ (1998) No. 24, 5243.
\bibitem{bender99}C. M. Bender, S. Boettcher and P. N. Meisinger, J. Math. Phys. $\bf{40}$ No.5 (1999).
\bibitem{bender02}C. M. Bender, D. C. Brody and H. F. Jones , Phys. Rev. Lett. $\bf{89}$ (2002) No. 27, 270401.
\bibitem{mosta02}A. Mostafazadeh, J. Math. Phys. $\bf{43}$ (2002) 205.
\bibitem{zheng13}C. Zheng, L. Hao, and G. L. Long, Philos. Trans. R. Soc.A$\bf{371}$ (2013) 20120053.
\bibitem{bitt12}S. Bittner, B. Dietz, U. G$\ddot{\rm u}$nther, H. L. Harney, M.
\bibitem{rubin07}J. Rubinstein, P. Sternberg and Q. Ma, Phys. Rev. Lett. $\bf{99}$ (2007) 167003.
\bibitem{graefe08}E. M. Graefe, U. G$\ddot{\rm u}$nther, H. J. Korsch, A. E. Niederle, J. Phys. A: Math. Theor. $\bf{41}$ (2008) 255206.
\bibitem{krei16}M. Kreibich, J. Main, H. Cartarius and G. Wunner, Phys. Rev. A. $\bf{93}$  (2016) 023624.
\bibitem{klai08}S. Klaiman, N. Moiseyev, U. G$\ddot{\rm u}$nther, Phys. Rev. Lett. $\bf{101}$ (2008) 080402.
\bibitem{brody17}D. C. Brody, J. Phys. A: Math. Theor. $\bf{50}$ (2017) 485202.
\bibitem{mosta05}A. Mostafazadeh, J. Math. Phys. $\bf{46}$ (2005) 102108.
\bibitem{beygi15}A. Beygi, S. P. Klevansky and C. M. Bender, Phys. Rev. A. $\bf{91}$ (2015) 062101.
\bibitem{chakra20} A. Chakraborty, {\it J Phys. A : Math. Theor.} {\bf 53},
485202 (2020).
\bibitem{chakra21} A. Chakraborty, {\it Int. J. Theor. Phys.} {\bf 60}, 3689 (2021).
\bibitem{bieden81}L. C. Biedenharn and J. D. Louck:{\it Angular Momentum in Quantum Physics Theory and Application} (chapter:5 page no. 212-215) ENCYCLOPEDIA IN MATHEMATICS AND ITS APPLICATIONS(vol-8) Adisson Wesley Publishing House(1981).
\bibitem{hai18} P. V. Hai,  M. Putinar,   {\it J. Diff. Equations}   {\bf 265}, 4213 (2018).

\bibitem{hai16}P. V. Hai, L. H. Khoi,      {\it J. Math. Anal. Appl.}  $\bf{433}$, 1757 (2016).

\bibitem{hai18c} P. V. Hai, L. H. Khoi ,    {\it Complex Variables and Elliptic Equations}  {\bf 63}, 391 (2018).

\end{thebibliography}
\end{document}